\documentclass[11pt]{article}

\usepackage{hyperref}

\usepackage{amsmath,amssymb,graphicx,amsthm,dsfont}
\usepackage{color,soul}
\usepackage{xspace}
\usepackage{cleveref}
\usepackage{mathtools}

\usepackage{bbm}
\usepackage{etex}
\usepackage{comment}

\usepackage{nicefrac}
\usepackage[neverdecrease]{paralist}
\usepackage{todonotes}
\usepackage{setspace}
\usepackage[matrix,arrow,ps]{xy}

\usepackage{caption}
\usepackage[noend,boxed,vlined,ruled,linesnumbered]{algorithm2e} 
\usepackage{xifthen}
\usepackage{tabularx}
\usepackage{makecell}
\usepackage{multirow}
\usepackage{mathrsfs}
\usepackage{url}
\usepackage[numbers]{natbib}
\usepackage{breakurl}
\usepackage[shortlabels]{enumitem}

\newcommand{\naturals}{{{\mathbb{N}}}}

\newcommand{\reals}{{{\mathbb{R}}}}

\newcommand{\calJ}{\mathcal{J}}
\newcommand{\calJS}{\mathcal{J_S}}

\newcommand{\calP}{\mathcal{P}}
\newcommand{\calR}{\mathcal{R}}
\newcommand{\calX}{\mathcal{X}}
\newcommand{\calF}{\mathcal{F}}

\newcommand{\calA}{\mathcal{A}}
\newcommand{\calT}{\mathcal{T}}
\newcommand{\dis}{\mathrm{dis}}

\newcommand\fixme[1]{}

\newcommand{\id}{{{\mathrm{id}}}}
\newcommand{\np}{{{\mathrm{NP}}}}
\newcommand{\pos}{{{{\mathrm{pos}}}}}

\newtheorem{theorem}{Theorem}

\newtheorem{definition}{Definition}
\newtheorem{example}{Example}
\newtheorem{proposition}[theorem]{Proposition}

\theoremstyle{plain}

\crefname{table}{Table}{Tables}
\crefname{figure}{Figure}{Figures}
\crefname{theorem}{Theorem}{Theorems}
\Crefname{theorem}{Thm.}{Thms.}
\crefname{definition}{Definition}{Definitions}
\crefname{corollary}{Corollary}{Corollaries}
\crefname{observation}{Observation}{Observations}
\crefname{lemma}{Lemma}{Lemmas}
\crefname{example}{Example}{Examples}
\crefname{reduction}{Reduction}{Reductions}
\crefname{construction}{Construction}{Constructions}
\crefname{subsection}{Subsection}{Subsections}
\crefname{section}{Section}{Sections}
\crefname{proposition}{Proposition}{Propositions}
\Crefname{proposition}{Prop.}{Props.}
\crefname{algorithm}{Algorithm}{Algorithms}

\makeatletter
\newcommand{\gettikzxy}[3]{%
  \tikz@scan@one@point\pgfutil@firstofone#1\relax
  \edef#2{\the\pgf@x}%
  \edef#3{\the\pgf@y}%
}
\makeatother
\renewenvironment{description}
  {\list{}{\labelwidth=10pt \leftmargin=15pt
   }}
  {\endlist}

\setlist[enumerate,1]{leftmargin=14pt}
\setlist[enumerate,2]{leftmargin=16pt}

\hypersetup{
colorlinks=true,citecolor=blue,linkcolor=red!60!black,
pagebackref=true
}

\newif\ifarxiv
\arxivfalse

\usepackage{comment}
\ifarxiv

\else
\excludecomment{hideproof}
\fi

\begin{document}
\sloppy
\allowdisplaybreaks

\title{Collective Schedules:\\Scheduling Meets Computational Social Choice}

\author{
  Fanny Pascual \\ 
Sorbonne Universit\'e, CNRS, LIP6\\
  Paris, France \\
  \texttt{fanny.pascual@lip6.fr}
\and 
  Krzysztof Rzadca\\
  University of Warsaw\\
  Warsaw, Poland\\
  \texttt{krz@mimuw.edu.pl}
\and
  Piotr Skowron\\
  University of Warsaw\\
  Warsaw, Poland\\
  \texttt{p.skowron@mimuw.edu.pl}
}

\maketitle

\begin{abstract}

  When scheduling public works or events in a shared facility
  one needs to accommodate preferences of a population.
  We formalize this problem by introducing the notion of a \emph{collective schedule}. We show how to extend fundamental tools from social choice theory---positional scoring rules, the Kemeny rule and the Condorcet principle---to collective scheduling.
  We study the computational complexity of finding collective schedules.
  We also experimentally demonstrate that optimal collective schedules can be found for instances with realistic sizes.
\end{abstract}

\section{Introduction}
Major public infrastructure projects, such as extending the city subway system, are often phased. As workforce, machines and yearly budgets are limited, phases have to be developed one by one.
Some phases are inherently longer-lasting than others.
Moreover, individual citizens have different preferred orders of phases. Should the construction start with a long phase with a strong support, or rather a less popular phase, that, however, will be finished faster?
If the long phase starts first, the citizens supporting the short phase would have to wait significantly longer.
Consider another example: planning events in a single lecture theater for a large, varied audience. The theater needs to be shared among different groups.
Some events last just a few hours, while others multiple days. What is the optimal schedule?
We formalize these and similar questions by introducing the notion of a \emph{collective schedule}, a plan that takes into account both jobs' durations and their societal support. 
The central idea stems from the observation that the problem of finding a socially optimal collective schedule is closely related to the problem of aggregating agents' preferences, one of the central problems studied in social choice theory~\citep{arrow2010handbook}.
However, differences in jobs' lengths have to be explicitly considered.
Let us illustrate these similarities through the following example.

Consider a collection of jobs all having the same duration.
The jobs have to be processed sequentially (one by one).
Different agents might have different preferred schedules of processing these jobs.
Since each agent would like all the jobs to be executed as soon as possible, the preferred schedule of each agent does not contain ``gaps'' (idle times), and so, such a preferred schedule can be viewed as an order over the set of jobs, and can be interpreted as a preference relation. Similarly, the resulting collective schedule can be viewed as an aggregated preference relation.
From this perspective, it is natural to apply tools from social choice theory to find a socially desired collective schedule.

Yet, the tools of social choice cannot be always applied directly. The scheduling model is typically much richer, and contains additional elements. In particular, when jobs' durations vastly differ, these differences must be taken into account when constructing a collective schedule. For instance, imagine that we are dealing with two jobs---one very short, $J_s$, and one very long, $J_l$. Further, imagine that 55\% of the population prefers the long job to be executed first and that the remaining 45\% has exactly opposite preferences. If we disregard the jobs' durations, then perhaps every decision maker would schedule $J_l$ before $J_s$.
However,
starting with $J_s$ affects 55\% of population just slightly (as $J_l$ is just slightly delayed compared to their preferred schedules). In contrast, starting with $J_l$ affects 45\% of population significantly (as $J_s$ is severely delayed).

\subsection{Overview of Our Contributions}
We explore the following question: How can we meaningfully apply the classic tools from social choice theory to find a collective schedule?
The key idea behind this work is to use fundamental concepts from both fields to highlight the new perspectives.

Scheduling offers an impressive collection of models, tools and algorithms which can be applied to a broad class of problems. 
It is impossible to cover all of them in a single work.
We use perhaps the most fundamental (although still non-trivial) scheduling model: a single processor executing a set of independent jobs.
This model is already rich enough to describe significant real-world problems (such as the public works or the lecture theater introduced earlier).
At the same time, such a model, fundamental, well-studied and stripped from orthogonal issues, enables us to highlight the new elements brought by social choice.

Similarly, we focus on three well-known and extensively studied tools from social choice theory: positional scoring rules, the Kemeny rule and the Condorcet principle.
Under \emph{a positional scoring rule} the score that an object receives from an agent is derived only on the basis of the position of this object in the agent's preference ranking; the objects are then ranked in the descending order of their total scores received from all the agents. 
The \emph{Kemeny rule} uses the concept of distances between rankings. It selects a ranking which minimizes the sum of the swap distances to the preference rankings of all the agents. The \emph{Condorcet principle} states that if there exists an object that is preferred to any other object by the majority of agents, 
then this object should be put on the top of the aggregated ranking.
The Condorcet principle can be generalized to the remaining ranking positions.
Assume that the graph of the preferences of the majority of agents is acyclic, i.e., there exists no such a sequence of objects $o_1, \ldots, o_{\ell}$ that $o_1$ is preferred by the majority of agents to $o_2$, $o_2$ to $o_3$, $\ldots$, $o_{\ell-1}$ to $o_{\ell}$ and $o_{\ell}$ to $o_1$.
Whenever an object $o$ is preferred by the majority of agents to another object $q$, $o$ should be put before $q$ in the aggregated ranking.   

Naturally, these three notions
can be directly applied to find a collective schedule. Yet, as we argued in our example with a long and a short job, this can lead to intuitively suboptimal schedules, because they do not consider significantly different processing times. 
We propose extensions of these tools to take into account lengths of the jobs. We also analyze their computational complexity.

\subsection{Related Work}
\textbf{Scheduling:} The two most related scheduling models apply concepts from game theory and multiagent optimization. The selfish job model~\citep{worstCaseEqulibria,vocking_algorithmic_2007} assumes that each job has a single owner trying to minimize its completion time and that the jobs compete for processors. The multi-organizational model~\citep{dutot_approximation_2011} assumes that a single organization owns and cares about multiple jobs. Our work complements these with a third perspective: not only each job has multiple ``owners'', but also they care about all jobs (albeit to a different degree).

In multiagent scheduling \citep{agnetis2014multiagent}, agents have different optimization goals (e.g., different functions or weights). The system's objective is to find all Pareto-optimal schedules, or a single Pareto-optimal schedule (optimizing one agent's goal with constraints on admissible values for other goals).
In contrast, 
our aim is to 
propose rules allowing to construct a single, compromise schedule.
This compromise stems from social choice methods and tools.
Moreover, our setting is motivated by problems in which the number of agents is large. To the best of our knowledge, the existing literature on multiagent scheduling focuses on cases with a few (e.g. two) agents.


\medskip
\noindent\textbf{Computational social choice:} For an overview of tools and methods for aggregating agents' preferences see the book of~\citet{arrow2010handbook}. \citet{fis:weighted-tournament} overview the computational complexity of finding Kemeny rankings. \citet{car:dodgson-young} discuss computational complexity of finding winners according to a number of 
Condorcet-consistent methods. 

Typically in social choice, an aggregated ranking is created to establish the collective preference relation, and to eventually select a single best alternative (sometimes with a few runner-ups). Thus, the agents usually do not care what is the order of the candidates in the further part of the collective ranking. In our model the agents are interested in the whole output rankings.
We can thus implement fairness---the agents who are dissatisfied with an order in the beginning of a collective schedule might be compensated in the further part of the schedule. Thus, our approach is closer to the recent works of~\citet{proprank} and~\citet{CSV17} analyzing fairness of collective rankings. 


In participatory budgeting \citep{cabannes:participatory, knapsackVoting, pbp, conf/wine/FainGM16, aaai/BenadeNP017} agents express preferences over projects which have different costs. The goal is to choose a socially-optimal set of items with a total cost not exceeding the budget. Thus, in a way, participatory budgeting extends the knapsack problem similarly to how we extend scheduling.

\section{The Collective Scheduling Model}\label{sec:model}
We use standard scheduling notations and definitions from the book of~\citet{brucker2006scheduling}, unless otherwise stated. For each integer $t$, by $[t]$ we denote the set $\{1, \ldots, t\}$.
Let $N = [n]$ be the set of $n$ agents (voters) and let $\calJ = \{J_1, \ldots, J_m\}$ be the set of $m$ jobs (note that in scheduling $m$ is typically used to denote the number of machines; we deliberately abuse this notation as our results are for a single machine).
For a job $J_i$ by $p_i \in \naturals$ we denote its processing time (also called duration or size), i.e., the number of time units $J_i$ requires to be completed.
We consider an off-line problem, i.e., jobs $\mathcal{J}$ are known in advance.
Jobs are ready to be processed (there are no release dates).
For each job $J_i$ its processing time $p_i$ is known in advance (\emph{clairvoyance}, a standard assumption in the scheduling theory).
Once started, a job cannot be interrupted until it completes (we do not allow for  \emph{preemption} of the jobs).

There is a single machine that executes all the jobs.
A schedule $\sigma\colon \calJ \to \naturals$ is a function that assigns to each job $J_i$  its start time $\sigma(J_i)$,
such that no two jobs $J_k, J_\ell$ execute simultaneously.
Thus, either $\sigma(J_k) \geq \sigma(J_\ell) + p_\ell$ or $\sigma(J_\ell) \geq \sigma(J_k) + p_k$.
By $C_i(\sigma)$ we denote the completion time of job $J_i$: $C_i(\sigma) = \sigma(J_i) + p_i$.
We assume that a schedule has no gaps: for each job $i$, except the job that completes as the last one, there exists job $j$ such that $C_i(\sigma)=\sigma(J_j)$.
Let $\mathscr{S}$ denote the set of all possible schedules for the set of jobs $\calJ$.

Each agent wants all jobs to be completed as soon as possible, yet agents differ in their views on the relative importance of the jobs.
We assume that each agent $a$ has a certain preferred schedule $\sigma_a \in \calJ$, and when building $\sigma_a$, an agent is aware of the processing times of the jobs.
In particular, $\sigma_a$ does not have to directly correspond to the relative importance of jobs. 
For instance, if in $\sigma_a$ a short job $J_s$ precedes a long job $J_{\ell}$, then this does not necessarily mean that $a$ considers $J_s$ more important than $J_{\ell}$.
$a$ might consider $J_{\ell}$ more important, but she might prefer a marginally less important job $J_s$ to be completed sooner as it would delay $J_{\ell}$ only a bit.

A schedule can be encoded as a
(transitive, asymmetric) binary relation: $J_i \; \sigma_a \; J_k \Leftrightarrow \sigma_a(J_i) < \sigma_a(J_k)$.
E.g., $J_1 \;\sigma_a\; J_2 \;\sigma_a\; \ldots \;\sigma_a\; J_m$ means that
agent~$a$ wants $J_1$ to be processed first, $J_2$ second, and so on. We will denote such a schedule as $(J_1, J_2, \ldots, J_m)$.

We call a vector of preferred schedules, one for each agent, a \emph{preference profile}. By $\mathscr{P}$ we denote the set of all preference profiles of the agents.
A \emph{scheduling rule} $\calR \colon \mathscr{P} \to \mathscr{S}$ is a function which takes a preference profile as an input and returns a collective schedule.




In the remaining part of this section we propose different methods in which
the preference profile is used to evaluate a proposed collective schedule $\sigma$ (and thus, to construct a scheduling rule~$\calR$). 
All the proposed methods extrapolate information from $\sigma_a$ (a preferred schedule) to evaluate $\sigma$. 
Such an extrapolation is common in social choice: in participatory budgeting it is typical to ask each agent to provide a single set of items~\cite{cabannes:participatory, knapsackVoting, pbp, aaai/BenadeNP017} (instead of preferences over sets of items); similarly in multiwinner elections, each agent provides separable preferences of candidates~\cite{sko-fal-lan:j:collective, FSST-trends}. Alternatively, we could ask an agent to express her preferences over all possible schedules. This approach is also common in other areas of social choice (e.g., in voting in combinatorial domains model~\cite{LangXia15}), yet it requires eliciting exponential information from the agents. There exist also middle ground approaches, using specifically designed languages, such as CP-nets, for expressing preferences.

\subsection{Scheduling by Positional Scoring Rules}\label{sec:psf_functions}

In the classic social choice, positional scoring rules are perhaps the most straightforward, and the most commonly used in practice, tools to aggregate agents' preferences. Informally, under a positional scoring rule each agent $a$ assigns a score to each candidate $c$ (a job, in our case), which depends only on the position of $c$ in $a$'s preference ranking. For each candidate the scores that she receives from all the agents are summed up, and the candidates are ranked in the descending order of their total scores.

There is a natural way to adapt this concept. For an increasing function $h \colon \naturals \to \reals$ and a job~$J$ we define the $h\text{-}\mathrm{score}$ of~$J$ as the total duration of jobs scheduled after $J$ in all preferred schedules:
\begin{align*}
h\text{-}\mathrm{score}(J) = \sum_{a \in N}f\left( \sum_{J_i\colon J \, \sigma_a \, J_i} p_i \right) \text{.}
\end{align*}

The $h$-psf-rule (psf for positional scoring function) schedules the jobs by their descending $h$-scores.
If jobs are unit-size ($p_i=1$), then $h\text{-}\mathrm{score}(J)$ is simply the score that $J$ would get from the classic positional scoring rule induced by $h$.
For an identity function $h_{\id}(x) = x$, 
the $h_{\id}$-psf-rule corresponds to the Borda voting method adapted to collective scheduling.

The so-defined scheduling methods differ from traditional positional scoring rules, by taking into account the processing times of the jobs:
\begin{enumerate}
\item A score that a job $J$ receives from an agent $a$ depends on the total processing time rather than on the number of jobs that $J$ precedes in schedule~$\sigma_a$.
\item When scoring a job~$J$ we sum the duration of jobs scheduled \emph{after} $J$, rather than before it. This implicitly favors jobs with lower processing times. Indeed, consider two preferred schedules, $\sigma$ and $\tau$ identical until time $t$, at which a long job $J_{\ell}$ is scheduled in $\sigma$, and a short job $J_{s}$ is scheduled in $\tau$. Since $J_{s}$ is shorter, the total size of the jobs succeeding $J_s$ in $\tau$ is larger than the total size of the jobs succeeding $J_{\ell}$ in $\sigma$. Consequently, $J_s$ gets a higher score from $\tau$ than $J_{\ell}$ gets from $\sigma$. 
\end{enumerate}

However, this implicit preference for short jobs seems insufficient, as illustrated by the following example.

\begin{example}\label{ex:psf_bad}
Consider three jobs, $J_{\ell,1}, J_{\ell,2}, J_{s}$, with the processing times $\ell$, $\ell$, and $1$, respectively. Assume that $\ell \gg 1$, and consider the following preferred schedules of agents:
\begin{align*}
\hspace{1.5cm}\nicefrac{3n}{8} + \epsilon \text{~of agents~} &\colon\quad J_{\ell,1} \; &&\sigma \; && J_{\ell,2} \; &&\sigma \; && J_{s} \\
\nicefrac{3n}{8} + \epsilon \text{~of agents~} &\colon\quad J_{\ell,2} \; &&\sigma \; && J_{\ell,1} \; &&\sigma \; && J_{s} \\
\nicefrac{n}{8} - \epsilon \text{~of agents~} &\colon\quad J_{s} \;          &&\sigma \; && J_{\ell,1} \; &&\sigma \; && J_{\ell,2} \\
\nicefrac{n}{8} - \epsilon \text{~of agents~} &\colon\quad J_{s} \;          &&\sigma \; && J_{\ell,2} \; &&\sigma \; && J_{\ell,1}
\end{align*}
\noindent By $h_{\id}$-psf-rule, $J_{\ell,1}$ and $J_{\ell,2}$ are scheduled before~$J_{s}$.
However, starting with $J_s$ would delay $J_{\ell,1}$ and $J_{\ell,2}$ by only one time unit, while starting with $J_{\ell,1}$ and $J_{\ell,2}$ delays $J_s$ by $2\ell$, an arbitrarily large value. Moreover, $J_s$ is put first by roughly $\nicefrac{1}{4}$ of agents, a significant fraction.
\end{example}

Example~\ref{ex:psf_bad} demonstrates that the pure social choice theory does not offer tools appropriate for collective scheduling (we will provide more arguments to support this statement throughout the text). 
To address such issues we propose an approach that builds upon social choice \emph{and} the scheduling theory.

\subsection{Scheduling Based on Cost Functions}\label{sec:cost_functions}
A cost function quantifies how a given schedule $\tau$ differs from an agent's preferred schedule $\sigma$.
In this section, we adapt to our model classic costs used in scheduling and in social choice. We then show how to aggregate these costs among agents in order to produce a single measure of a quality of a schedule. This approach allows us to construct a family of scheduling methods that, in some sense, extend the classic Kemeny rule.

Formally, a cost function $f$
maps a pair of schedules, $\tau$ and $\sigma$, to a non-negative real value.
We analyze the following cost functions.
Below, $\tau$ denotes a collective schedule the quality of which we want to assess; while $\sigma$ denotes the preferred schedule of a single agent.

\subsubsection{Swap Costs.}
These functions take into account only the orders of jobs in the two schedules (ignoring the processing times), thus directly correspond to costs from social choice.
    \begin{enumerate}
\item The \citet{Kendall1938measure} tau (or swap) distance (\textbf{K}), measures the number of swaps of adjacent jobs to turn one schedule into another one. We use an equivalent definition that counts all pairs of jobs executed in a non-preferred order:
 \begin{align*}
  \textstyle K(\tau, \sigma) = \Big| \big\{ (k, \ell) \colon J_k \; \tau \; J_\ell \; \text{and} \; J_\ell \; \sigma \; J_k \big\} \Big| \text{.}
 \end{align*}
\item Spearman distance (\textbf{S}). Let $\pos(J, \pi)$ denote the position of job $J$ in a schedule $\pi$, i.e., the number of jobs scheduled before $J$ in $\pi$. The Spearman distance is defined as:
 \begin{align*}
 \textstyle S(\tau, \sigma) = \sum_{J \in \calJ} \big|\pos(J, \sigma) - \pos(J, \tau)\big| \text{.}
 \end{align*}
\end{enumerate}

\subsubsection{Delay Costs.}
These functions use the \emph{completion times} $\{ C_i(\sigma)\colon J_i \in \calJ \}$ of jobs in the preferred schedule $\sigma$ (and thus, indirectly, jobs' lengths). The completion times form jobs' due dates, $d_i=C_i(\sigma)$. A delay cost then quantifies how far are the proposed completion times $\{c_i = C_i(\tau)\colon J_i \in \calJ \}$ from their due dates $\{ d_i \}$ by one of the six classic criteria defined in~\citet{brucker2006scheduling}:
\begin{description}
  \item[Tardiness (\textbf{T})] $T(c_i, d_i)= \max (0, c_i - d_i)$.
  \item[Unit penalties (\textbf{U})] how many jobs are late:
  \begin{align*}
  U(c_i, d_i) =
  \begin{cases}
    1       & \quad \text{if } c_i > d_i\\
    0       & \quad \text{otherwise.}
  \end{cases}
  \end{align*}
  \item[Lateness (\textbf{L})] is similar to tardiness, but includes a bonus for being early: $L(c_i, d_i) = c_i - d_i$.
  \item[Earliness (\textbf{E})] $E(c_i, d_i) = \max (0, d_i - c_i )$.
  \item[Absolute deviation (\textbf{D})] $D(c_i, d_i) = |c_i - d_i|$.
  \item[Squared deviation (\textbf{SD})] $\mathit{SD}(c_i, d_i) = (c_i - d_i)^2$.
\end{description}
\smallskip


Each such a criterion $f \in \{T, U, L, E, D, \mathit{SD}\}$ naturally induces the corresponding delay cost of an agent, $f(\tau, \sigma)$: 
\begin{align*}
f(\tau, \sigma) = \sum_{J_i \in \calJ}f\Big(C_i(\tau), C_i(\sigma)\Big) \text{.}
\end{align*}

In this work, we mostly focus on the tardiness $T$, which is both easy to interpret for our motivating examples and the most extensively studied in scheduling. However, there is interest to study the remaining functions as well.
$U$ and $L$ are similar to $T$---the sooner a task is completed, the better.
The remaining three measures ($E, S$, and $\mathit{SD}$) penalize the jobs which are executed before their ``preferred times''. However, each job when executed earlier makes other jobs executed later (e.g., after their due times). 
  Thus, these penalties quantify the unnecessary (wasted) promotion of jobs executed too early (causing other jobs being executed too late).\footnote{The considered metrics have their natural interpretations also in other more specific settings. E.g., the earliness $E$ is useful if each task represents a (collective) work to be done by the agents (workers) and when agents do not want to work before their preferred start times. Similarly, $D$ and $\mathit{SD}$ can be used when an agent wants each task to be executed exactly at the preferred time.}
By restricting the instances to unit-size jobs, 
we can relate delay and swap costs.
The Spearman distance $S$ has the same value as the absolute deviation $D$ (by definition), and twice that of $T$: 
\begin{proposition}\label{prop:TandS}
For unit-size jobs it holds that $S(\sigma, \tau) = 2T(\sigma, \tau)$, for all schedules $\sigma, \tau$.
\end{proposition}

\begin{proof}
Observe that for unit-size jobs the tardiness measure can be expressed as:
\begin{align*}
T(\tau, \sigma) = \sum_{\mathclap{J\colon \pos(J, \tau) > \pos(J, \sigma)}} \Big(\pos(J, \tau) - \pos(J, \sigma)\Big) \text{,}
\end{align*}
Since $\sum_{J} \pos(J, \tau) = \sum_{J} \pos(J, \sigma)$ we get that:
\begin{align*}
\; 0 &= \sum_{J} \Big(\pos(J, \tau) - \pos(J, \sigma)\Big) \\
  &=  
    \sum_{\mathclap{J\colon \pos(J, \tau) > \pos(J, \sigma)}} \Big(\pos(J, \tau) - \pos(J, \sigma)\Big) + \sum_{\mathclap{J\colon \pos(J, \tau) < \pos(J, \sigma)}} \Big(\pos(J, \tau) - \pos(J, \sigma)\Big) +\\
  &\; + \sum_{\mathclap{J\colon \pos(J, \tau) = \pos(J, \sigma)}} \Big(\pos(J, \tau) - \pos(J, \sigma)\Big) \\
  &= \sum_{\mathclap{J\colon \pos(J, \tau) > \pos(J, \sigma)}} \Big(\pos(J, \tau) - \pos(J, \sigma)\Big) + \sum_{\mathclap{J\colon \pos(J, \tau) < \pos(J, \sigma)}} \Big(\pos(J, \tau) - \pos(J, \sigma)\Big) \text{.}
\end{align*}
Thus:
\begin{align*}
\quad\; \sum_{\mathclap{J\colon \pos(J, \tau) > \pos(J, \sigma)}} \Big|\pos(J, \tau) - \pos(J, \sigma)\Big| = \sum_{\mathclap{J\colon \pos(J, \tau) < \pos(J, \sigma)}} \Big|\pos(J, \tau) - \pos(J, \sigma)\Big| \text{.}
\end{align*}
And, consequently:
\begin{align*}
 S(\tau, \sigma) &= \sum_{J} \Big|\pos(J, \tau) - \pos(J, \sigma)\Big| \\
                 &= \sum_{\mathclap{J\colon \pos(J, \tau) > \pos(J, \sigma)}} \Big(\pos(J, \tau) - \pos(J, \sigma)\Big) + \\
  &\; + \sum_{\mathclap{J\colon \pos(J, \tau) < \pos(J, \sigma)}} \Big(\pos(J, \tau) - \pos(J, \sigma)\Big)\\
                 &= 2 \sum_{\mathclap{J\colon \pos(J, \tau) > \pos(J, \sigma)}} \Big(\pos(J, \tau) - \pos(J, \sigma)\Big) = 2T(\tau, \sigma) \text{.}
\end{align*}
This completes the proof.
\end{proof}

Since different agents can have different preferred schedules, in order to score a proposed schedule $\tau$ we need to aggregate the costs across all agents.
We will consider three classic aggregations:
\begin{description}
\item[The sum ($\Sigma$):] $\sum_{a \in N} f(\tau, \sigma_a)$, a utilitarian aggregation.
\item[The max:] $\max_{a \in N} f(\tau, \sigma_a)$, an egalitarian 
  aggregation.
\item[The $L_p$ norm ($L_p$):]
  $\sqrt[p]{\sum_{a \in N} \big(f(\tau, \sigma_a)\big)^p}$,
  with a parameter~\mbox{$p \geq 1$}. The $L_p$ norms form a spectrum of aggregations between the sum ($L_1$) and the max ($L_\infty$).
\end{description}

For a cost function $f \in \{K, S, T, U, L, E, D, \mathit{SD}\}$ and an aggregation $\alpha \in \{\Sigma, \max, L_p\}$, by $\alpha$-$f$ we denote a scheduling rule returning a schedule that minimizes the $\alpha$-aggregation of the $f$-costs of the agents. In particular, for unit-size jobs the $\Sigma$-$T$ rule is equivalent to $\Sigma$-$S$ and to $\Sigma$-$D$, and $\Sigma$-$K$ is simply the Kemeny rule.

Scheduling based on cost functions avoids the problems exposed by Example~\ref{ex:psf_bad} (indeed for that instance, e.g., the $\Sigma$-$T$ rule starts with the short job $J_s$). Additionally, these methods satisfy some naturally-appealing axiomatic properties, such as reinforcement, which is a particularly natural requirement in our case.

\begin{definition}[Reinforcement]
A scheduling rule $\calR$ satisfies reinforcement iff for any two groups of agents $N_1$ and $N_2$, a schedule $\sigma$ is selected by $\calR$ both for $N_1$ and for $N_2$, then it should be also selected for the joint instance $N_1 \cup N_2$.  
\end{definition}

\begin{proposition}
All $\Sigma$-$f$ scheduling rules satisfy reinforcement.
\end{proposition}

\subsection{Beyond Positional Scoring Rules and Cost Functions: the Condorcet Principle}

In the previous section we introduced several scheduling rules, all based on the notion of a distance between schedules. Thus, these scheduling rules are closely related to the Kemeny voting system.
We now 
take a different approach. We start from desired properties of a collective schedule and design scheduling rules satisfying them.

Pareto efficiency is one of the most accepted axioms in social choice theory.
Below we use a formulation analogous to the one used in voting theory (based on swaps in preferred schedules).
\begin{definition}[Pareto efficiency]
A scheduling rule $\calR$ satisfies Pareto efficiency iff for each pair of jobs, $J_k$ and $J_\ell$, and for each preference profile $\sigma = (\sigma_1, \ldots, \sigma_n) \in \mathscr{P}$ such that for each $a \in N$ we have $J_k \; \sigma_a \; J_\ell$, it holds that $J_k \; \calR(\sigma) \; J_\ell$.  
\end{definition}
In other words, if all agents prefer $J_k$ to be scheduled before $J_\ell$, then in the collective schedule $J_k$ should be before $J_\ell$. Curiously, the total tardiness $\Sigma$-$T$ rule does not satisfy Pareto efficiency:

\begin{example}\label{ex:condorcet-paradox}
Consider an instance with 3 jobs $J_1, J_2, J_3$ with lengths 20, 5, and 1, respectively, and with two agents having preferred schedules $\sigma_a = (J_1, J_3, J_2)$ and $\sigma_b = (J_2, J_1, J_3)$.
Both agents prefer $J_1$ to be scheduled before $J_3$. If our scheduling rule satisfied
Pareto efficiency, then it would pick one of the following three schedules: $(J_1, J_3, J_2)$, $(J_1, J_2, J_3)$, or $(J_2, J_1, J_3)$. The total tardinesses of these schedules are equal to: 21, 25, and 10, respectively. Yet, the total tardiness of the schedule $(J_2, J_3, J_1)$ is equal to 7.
\end{example}

This example can be generalized to inapproximability:

\begin{proposition}\label{thm:tardiness_pareto_approx}
  For any $\alpha>1$, there is no scheduling rule that satisfies Pareto efficiency and is $\alpha$-approximate for $\max$-$T$ or $\Sigma$-$T$.
\end{proposition}
\begin{proof}
Let us assume, towards a contradiction, that there exists a scheduling rule $\calR$ that satisfies Pareto efficiency and is $\alpha$-approximate for minimizing $\Sigma$-$T$ (the proof for $\max$-$T$ is analogous). Let $x=\lceil 3\alpha \rceil$. Consider an instance with $x+2$ jobs: one job $J_1$ of length $x^2$, one job $J_2$ of length $x$, and $x$ jobs $J_3,\dots, J_{x+2}$  of length 1. Let us consider two agents with preferred schedules $\sigma_1 = (J_1, J_3, \dots, J_{x+2}, J_2)$ and $\sigma_2 = (J_2, J_1, J_3, \dots, J_{x+2})$. For each $i\in\{3,\dots,x+2\}$, both agents prefer job $J_1$ to be scheduled before job $J_i$. Let $\tau$ be the schedule returned by $\calR$.  Since $\calR$ satisfies Pareto efficiency, for each $i\in\{3,\dots,x+2\}$,  $J_1$ is scheduled before job $J_i$ in $\tau$.  Thus  $\tau$  is either $\sigma_2$  , or a schedule where $J_1$ is scheduled first, followed by $i$  jobs of length  1 ($i\in\{0,\dots,x\}$), followed by $J_2$, followed by the $x-i$ remaining jobs of length 1.  Let $S_{i}$ be such a schedule. In  $S_{i}$,  the tardiness of job $J_2$ is $x^2+i$ (this job is in first position in $\sigma_2$), and the tardiness of the jobs of length 1 is  $(x-i)x$ (the $x-i$ last jobs in $S_i$ are scheduled before $J_2$ in $\sigma_1$). Thus the total tardiness of $S_i$ is $(x^2+i)+(x-i)x\geq x^2+x$. The total tardiness of schedule $\sigma_2$ is $x^2+x$ (each of the $x$ jobs $J_1, J_3, \dots, J_{x+2}$ in $\sigma_2$ finishes $x$ time units later than in $\sigma_1$). Thus, the total tardiness of  $\tau$ is at least  $x^2+x$. Let us now consider schedule $\tau'$, which does not satisfy Pareto efficiency, and which is as follows: job $J_2$ is scheduled first, followed by the jobs of length 1, followed by job $J_1$. The total tardiness of this schedule is $3x$ (the only job which is delayed compared to $\sigma_1$ and $\sigma_2$ is job $J_1$).  This schedule is optimal for $\Sigma$-$T$. Thus the approximation ratio of  $\calR$ is at least $\frac{x^2+x}{3x}=\frac{x+1}{3}>\alpha$. Therefore, $\calR$ is not $\alpha$-approximate for $\Sigma$-$T$, a contradiction. 
\end{proof}

\begin{proposition}\label{thm:tardiness_unit_sizes_Pareto}
If all jobs are unit-size, the scheduling rule $\sum$-$T$ is Pareto efficient.
\end{proposition}
\begin{proof}
Let us assume that there exist two jobs which are not in a Pareto order in the schedule $\sigma$ optimizing $\sum T$. We can swap these jobs in $\sigma$ and it is apparent that such a swap does not increase the total tardiness of the schedule. We can perform such swaps until we reach a schedule which does not violate Pareto efficiency.
\end{proof}

Pareto efficiency is one of the most fundamental properties in social choice. However, sometimes (especially in our setting) there exist reasons for violating it. For instance, even if all the agents agree that $J_x$ should be scheduled before $J_y$, the preferences of the agents with respect to other jobs might differ. Breaking Pareto efficiency can help to achieve a compromise with respect to these other jobs.

Nevertheless, Proposition~\ref{thm:tardiness_pareto_approx} motivated us to formulate alternative scheduling rules based on axiomatic properties. We choose the Condorcet principle, a classic social choice property that is stronger than Pareto efficiency. We adapt it to consider the durations of jobs.


\begin{definition}[Processing Time Aware (PTA) Condorcet principle]\label{def:pta_condorcet}
A schedule $\tau \in \mathscr{S}$ is PTA Condorcet consistent with a preference profile $\sigma = (\sigma_1, \ldots, \sigma_n) \in \mathscr{P}$ if for each two jobs, $J_k$ and $J_\ell$, it holds that $J_k \; \tau \; J_\ell$ whenever at least $\frac{p_k}{p_k + p_\ell} \cdot n$ agents put $J_k$ before $J_\ell$ in their preferred schedule. A scheduling rule $\calR$ satisfies the PTA Condorcet principle if for each preference profile it returns a PTA Condorcet consistent schedule, whenever such exists.
\end{definition} 

Let us explain our motivation for ratio $\frac{p_k}{p_k + p_\ell}$.
Consider a schedule $\tau$ and two jobs, $J_k$ and $J_\ell$, scheduled consecutively in $\tau$. By $N_k$ we denote the set of agents who rank $J_k$ before $J_\ell$ in their preferred schedules, and let us assume that $|N_k| > \frac{p_k}{p_k + p_\ell}n$; we set 
$N_\ell = N - N_k$. 
Observe that if we swapped $J_k$ and $J_\ell$ in $\tau$, then each agent from $N_k$ would be disappointed.
Since such a swap
makes $J_k$ scheduled $p_\ell$ time units later than in $\tau$,
the level of dissatisfaction of each agent from $N_k$ could be quantified by $p_\ell$.
Thus, their total (utilitarian) dissatisfaction $\dis(N_k)$ could be quantified by $|N_k| \cdot p_{\ell}$. By an analogous argument, if we started with a schedule where $J_\ell$ is put right before $J_k$, and swapped these jobs, then the total dissatisfaction of agents from $N_\ell$ could be quantified by:
\begin{align*}
\dis(N_\ell) &= |N_\ell|p_k < \left(n - \frac{p_k}{p_k + p_\ell}n\right)p_k \\
          &= n\cdot \frac{p_k p_{\ell}}{p_k + p_\ell} < |N_k| \cdot p_{\ell} = \dis(N_k) \text{.}
\end{align*}
Thus, the total dissatisfaction of all agents from scheduling $J_k$ before $J_\ell$ is smaller than that from scheduling $J_\ell$ before $J_k$. Definition~\ref{def:pta_condorcet} requires that in such case $J_k$ should be indeed scheduled before $J_\ell$. 

Proposition~\ref{thm:tardiness_unit_sizes_PTA} below highlights the difference between scheduling based on the tardiness and on the PTA Condorcet principle. 

\begin{proposition}\label{thm:tardiness_unit_sizes_PTA}
Even if all jobs are unit-size, the  $\sum$-$T$ rule does not satisfy the PTA Condorcet principle.
\end{proposition}
\begin{proof}
Consider an instance with three jobs and three agents with the following preferred schedules:
\begin{align*}
&\sigma_1 = (J_1, J_2, J_3); \quad &\sigma_2 = (J_1, J_3, J_2); \quad &\sigma_3 = (J_1, J_3, J_2);  \\
&\sigma_4 = (J_2, J_3, J_1); \quad &\sigma_5 = (J_2, J_3, J_1). \quad &
\end{align*}
The only PTA Condorcet consistent schedule is $(J_1, J_2, J_3)$ with the total tardiness of 6. At the same time, the schedule $(J_1, J_3, J_2)$ has the total tardiness equal to 5.
\end{proof}

To construct a PTA Condorcet consistent schedule, we propose to extend Condorcet consistent~\cite{Col13, LevNal95} election rules to jobs with varying lengths. For example, we obtain:
\begin{description}
\item[PTA Copeland's method.] For each job $J_k$ we define the score of $J_k$ as the number of jobs $J_\ell$ such that at least $\frac{p_k}{p_k + p_\ell} \cdot n$ agents put $J_k$ before $J_\ell$ in their preferred schedule. The jobs are scheduled in the descending order of their scores.
\item[Iterative PTA Minimax.] For each pair of jobs, $J_k$ and $J_\ell$, we define the defeat score of $J_k$ against $J_\ell$ as $\max(0, \frac{p_k}{p_k + p_\ell}n - n_k)$, where $n_k$ is the number of agents who put $J_k$ before $J_\ell$ in their preferred schedule. We define the defeat score of  $J_k$ as the highest defeat score of $J_k$ against any other job. The job with the lowest defeat score is scheduled first. Next, we remove this job from the preferences of the agents, and repeat (until there are no jobs left).
\end{description}
Other Condorcet consistent election rules, such as the Dogdson's rule or the Tideman's ranked pairs method, can be adapted similarly. It is apparent that they satisfy the PTA Condorcet principle.

PTA Condorcet consistency comes at a cost: e.g., the two scheduling rules violate reinforcement, even if the jobs are unit-size.
Indeed, by the classic result of Young and Levenglick~\cite{lev-you:j:condorcet} one can infer that any rule that satisfies PTA-Condorcet principle, neutrality, and reinforcement must be a generalization of the Kemeny rule (i.e., must be equivalent to the Kemeny rule if the processing times of the jobs are equal). We conjecture that rules satisfying neutrality and reinforcement fail the PTA-Condorcet principle; it is an interesting open question whether such an impossibility theorem holds. 

\section{Computational Results}

In this section we study the computational complexity of finding collective schedules according to the previously defined rules. We start from the simple observation about the two PTA Condorcet consistent rules that we defined in the previous section.

\begin{proposition}
The PTA Copeland's method and the iterative PTA minimax rule are computable in polynomial time.
\end{proposition} 

We further observe that 
computational complexity of the rules which ignore the lengths of the jobs (rules based on swap costs) can be directly inferred from the known results from computational social choice. For instance, the $\Sigma$-$K$ rule is simply the well-known and extensively studied Kemeny rule. Thus, in the further part of this section we focus on the rules based on delay costs.

\subsection{Sum of Delay Costs}\label{sec:tard-sum}

First, observe that the problem of finding a collective schedule is computationally easy for the total lateness ($\Sigma$-$L$). In fact, $\Sigma$-$L$ ignores the preferred schedules of the agents and arranges the jobs from the shortest to the longest one. 

\begin{proposition}\label{prop:lateness_easy}
The rule $\Sigma$-$L$ schedules the jobs in the ascending order of their lengths.
\end{proposition}
\begin{proof}
Consider the total cost of the agents:  
\begin{align*}
\sum_{a \in N} L(\tau, \sigma_a) &= \sum_{a \in N} \sum_{J_i \in \calJ}(C_i(\tau) - C_i(\sigma_a)) 
    = |N| \sum_{J_i \in \calJ} C_i(\tau) - \sum_{a \in N} \sum_{J_i \in \calJ}C_i(\sigma_a) \text{.}
\end{align*}
Thus, the total cost of the agents is minimized when $\sum_{J_i \in \calJ} C_i(\tau)$ is minimal. This value is minimal when the jobs are scheduled from the shortest to the longest one. 
\end{proof}

On the other hand, minimizing the total tardiness $\Sigma$-$T$ is $\np$-hard even with the unary representation of the durations of jobs. \citet{du1990minimizing} show that minimizing total tardiness with \emph{arbitrary} due dates on a single processor ($1||\sum T_i$) is weakly $\np$-hard. We cannot use this result directly as the due dates in our problem $\Sigma$-$T$ are structured and depend, among others, on jobs' durations.

\begin{theorem}\label{thm:tardiness_np_hard}
The problem of finding a collective schedule minimizing the total tardiness ($\Sigma$-$T)$ is strongly $\np$-hard.
\end{theorem}
\begin{proof}
  We reduce from the strongly $\np$-hard \textsc{3-Partition} problem. Let $I$ be an instance of \textsc{3-Partition}.
  In $I$ we are given a multiset of integers $S = \{s_1, \ldots, s_{3\mu}\}$. We denote $s_{\Sigma} = \sum_{s \in S}s$.
  We ask if $S$ can be partitioned into $\mu$ triples that all have the same sum, $s_T = s_{\Sigma}/\mu$.
  Without loss of generality, we can assume that $\mu \geq 2$ and that for each $s \in S$, $\mu < s < \frac{s_{T}}{2}$ (otherwise, we can add a large constant $s_{\Sigma}$ to each integer from $S$, which does not change the optimal solution of the instance, but which ensures that $\mu < s < \frac{s_{T}}{2}$ in the new instance).
  We also assume that the integers from $S$ are represented in unary encoding.

  From $I$ we construct an instance $I'$ of the problem of finding a collective schedule that minimizes the total tardiness in the following way. For each number $s \in S$ we introduce $1+s\mu$ jobs: $J_s$ and $\big\{P_{s,i, j}\colon i \in [s], j \in [\mu] \big\}$.
We set the processing time of $J_s$ to $s$. Further, for each $i \in [s]$ we set the processing time of $P_{s, i, 1}$ to $(s_{T} - s)$, and of the remaining $j \geq 2$ jobs $P_{s, i, j}$ to $s_{T}$.
We denote the set of all such jobs as $\calJS = \{J_s\colon s \in S\}$ and $\calP = \big\{ P_{s,i,j} \colon s, i \in [s], j \in [\mu] \big\}$.
Additionally, we introduce $\mu$ jobs, $\calX = \{ X_1, \ldots, X_\mu \}$, each having a unit processing time.

There are $s_{\Sigma}$ agents. For each integer $s \in S$ we introduce $s$ agents. The $i$-th agent corresponding to number $s$, denoted by $a_{s,i}$, has the following preferred schedule (in the notation below a set, e.g., $\{ J_{s'} \}$ denotes that its elements are scheduled in a fixed arbitrary order):
\begin{align*}
  \Big(J_s, P_{s,i,1}, X_1, P_{s,i,2}, X_2, \ldots, P_{s,i,\mu}, X_\mu, \{J_{s'}\colon s' \neq s \}, 
  \big\{P_{s',j, \ell}\colon (s' \neq s \; \text{or} \; j \neq i) \; \text{and} \; \ell \in [\mu] \big\} \Big) \text{.}
\end{align*}

\begin{figure}[tb!]
  \centering
\includegraphics[width=\columnwidth]{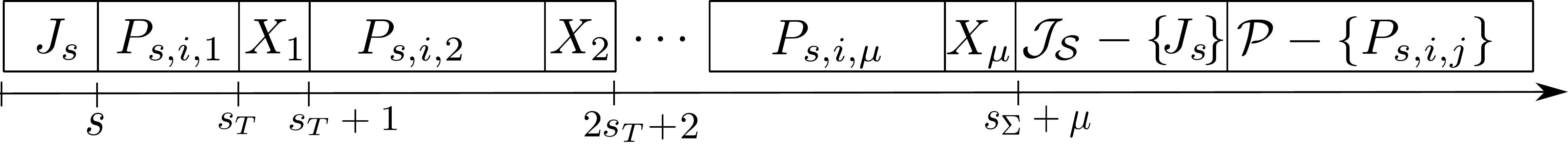}
 \includegraphics[width=\columnwidth]{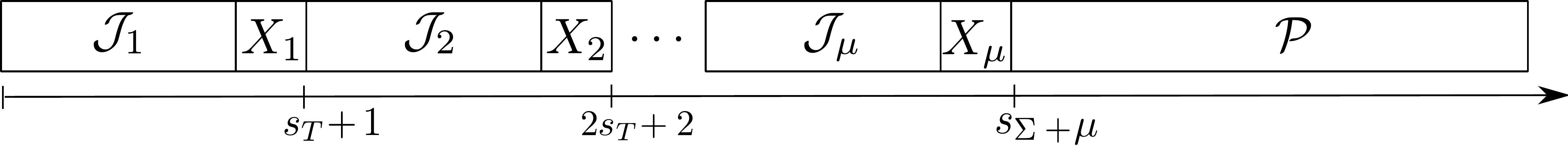}
  \caption{The preferred schedule $\sigma_{(s,i)}$ of agent $a_{s,i}$ (top) and the optimal schedule (bottom).}\label{fig:np-hard}
\end{figure}

We claim that the answer to the initial instance $I$ is ``yes'' if and only if the schedule $\sigma^*$ optimizing the total tardiness is the following one: $\Big(\calJ_1, X_1, \calJ_2, X_2, \calJ_\mu, X_\mu, \calP \Big)$, where for each $i \in [\mu]$, $\calJ_i$ is a set consisting of jobs from $\calJS$ with lengths summing up to $s_T$ (see Figure~\ref{fig:np-hard}). 
If such a schedule exists, then the answer to $I$ is ``yes''. Below we will prove the other implication.

Observe that any job from $\calJS$ should be scheduled before each job from $\calP$.
Indeed, for each pair $P_{s,i,j}$ 
and $J_{s'}$ only a single agent $a=a_{s,i}$ ranks $P_{s,i,j}$ before $J_{s'}$; at the same time there exists another agent $a'=a_{s', k}$ who ranks $J_{s'}$ first.
As $J_{s'}$ is shorter than $P_{s,i,j}$,
$a'$ gains more from $J_{s'}$ scheduled before $P_{s,i,j}$,
than $a$ gains from  $P_{s,i,j}$ scheduled before $J_{s'}$.
Thus, if $P_{s,i,j}$ were scheduled before $J_{s'}$, we could swap these two jobs and improve the schedule (such a swap could only improve the completion times of other jobs since $J_{s'}$ is shorter than $P_{s,i,j}$).  

By a similar argument, any job from $\calX$ should be scheduled before each job from $\calP$. Indeed, if it was not the case, then there would exist jobs $P = P_{s,i,j}$ and $X = X_{i'}$ such that $P$ is scheduled right before $X$ (this follows from the reasoning given in the previous paragraph---a job from $\calJS$ cannot be scheduled after a job from $\calP$). Also, since all the jobs from $\calJS$ are scheduled before $P$, the completion time of $X$ would be at least
$s_{\Sigma} + \frac{s_T}{2} + 1 \geq s_{\Sigma} + \mu + 2$.
For each agent, the completion time of $X$ in their preferred schedule is at most equal to 
$\mu(s_T + 1)=s_{\Sigma}+\mu$.
Thus, if we swap $X$ and $P$ the improvement of the tardiness due to scheduling $X$ earlier would be at least equal to 
$2 s_{\Sigma}$.
Such a swap increases the completion time of $P$ only by one, so the increase of the tardiness due to scheduling $P$ later would be at most equal to $s_{\Sigma}$. Consequently, a swap would decrease the total tardiness, and so $X$ could have not been scheduled after $P$ in $\sigma^*$.   
 
We further investigate the structure of an optimal schedule $\sigma^*$. We know that $\calJS \; \sigma^* \; \calP$ and that $\calX \; \sigma^* \; \calP$, but we do not yet know the optimal order of jobs from $\calJS \cup \calX$. Before proceeding further, we introduce one useful class of schedules, $\calT$, that execute jobs in the order $( \calJS, \calX, \calP )$. Observe that $\sigma^*$ can be constructed starting from some schedule $\tau \in \calT$ and performing a sequence of swaps, each swap involving a job $J \in \calJS$ and a job $X \in \calX$.
The tardiness of $\sigma^*$ is equal to the tardiness of the initial $\tau$ adjusted by the changes due to the swaps.
Below, we further analyze $\calT$.
First, any ordering of $\calJS$ in $\tau$ results in the same tardiness.
Indeed, consider two jobs $J_s$ and $J_{s'}$ such that $J_{s'}$ is scheduled right after $J_{s}$. If we swap $J_{s}$ and $J_{s'}$, then the total tardiness of $s$ agents increases by $s'$ and the total tardiness of $s'$ agents decreases by $s$. In effect, the total tardiness of all agents remains unchanged.
Second, there exists an optimal schedule where the relative order of the jobs from $\calX$ is $X_1 \, \sigma^* \, X_2 \, \sigma^* \ldots \sigma^* \, X_\mu$. Thus, w.l.o.g., we constrain $\calT$ to schedules in which $\calX$ are put in exactly this order. 

Since we have shown that all $\calT$ always have the same tardiness, no matter how we arrange the jobs from $\calJS$, the tardiness of $\sigma^*$ only depends on the change of the tardiness due to the swaps. Consider the job $X_1$, and consider what happens if we swap $X_1$ with a number of jobs from $\calJS$ so that eventually $X_1$ is scheduled at time $s_T$ (its start time in all preferred schedules). In such a case, moving $X_1$ forward decreases the tardiness of each of $s_\Sigma$ agents by $(s_\Sigma - s_T)$. Moving $X_1$ forward to $s_T$ requires however delaying some jobs from $\calJS$.  Assume that the jobs from $\calJS$ with the processing times $s_{i_1}, \ldots s_{i_{\ell}}$ are delayed. Each such job needs to be scheduled one time unit later. Thus, the total tardiness of $s_{i_1}$ agents increases by 1 (the agents who had this job as the first in their preferred schedule), of other $s_{i_2}$ agents increases by 1, and so on. Since $s_{i_1} + \ldots + s_{i_{\ell}} = s_{\Sigma}-s_T$, the total tardiness of all agents increases by $s_{\Sigma}-s_T$.
Thus, in total, executing $X_1$ at $s_T$ decreases the total tardiness by $s_{\Sigma} ( s_\Sigma - s_T) - ( s_{\Sigma} - s_T )$, a positive number. Also, observe that this value does not depend on how the jobs from $\calJS$ were initially arranged, provided that $X_1$ can be put so that it starts at $s_T$.

Starting $X_1$ earlier than $s_T$ does not improve the tardiness of $X_1$, yet it increases tardiness of some other jobs, so it is suboptimal. By repeating the same reasoning for $X_2, \ldots, X_\mu$ we infer that we obtain the optimal decrease of the tardiness when $X_1$ is scheduled at time $s_T$, $X_2$ at time
$2 s_T + 1$, etc., and if there are no gaps between the jobs.
However, such schedule is possible to obtain if and only if the answer to the initial instance of 3-Partition is ``yes''.
\end{proof}

A similar strategy (yet, with a more complex  construction) can be used to prove the $\np$-hardness of $\Sigma$-$U$. 

\begin{theorem}\label{thm:unit_penalties_hard}
The problem of finding a collective schedule minimizing the total number of late jobs ($\Sigma$-$U$) is strongly $\np$-hard.
\end{theorem}
\begin{proof}
  We give a reduction from the strongly $\np$-hard \textsc{3-Partition} problem. Let $I$ be an instance of \textsc{3-Partition}.
  In $I$ we are given a multiset of $3\mu$ integers $S = \{s_1, \ldots, s_{3\mu}\}$. Similarly, as in the proof of Theorem~\ref{thm:tardiness_np_hard}, we set $s_{\Sigma} = \sum_{s \in S}s$.
  In $I$ we ask if $S$ can be partitioned into $\mu$ triples that all have the same sum, $s_T = s_{\Sigma}/\mu$.
  We assume that for each $s \in S$, $s < \frac{s_T}{2}$, that $\mu>4$, and that the integers from $S$ are represented in unary encoding.

  From $I$ we construct an instance $I'$ of the problem of finding a collective schedule that minimizes the total number of late jobs in the following way. For each number $s \in S$ we introduce the following jobs:
\begin{itemize}
\item a job $F_s$ of length $s$; 
\item $s \mu$ jobs of length $s_T-s$; we denote this set as:
\begin{align*}
\mathcal{R}_s=\big\{ R_{s,i,j} \colon s, i \in [s] j \in [\mu]\big\};
\end{align*}
\item $\mu(\mu-1)s$ jobs of length $s_T$; we denote this set as:
\begin{align*}
\calP_s = \big\{ P_{s,i,j,k} \colon s, i \in [s], j \in [\mu], k \in [\mu -1] \big\}.
\end{align*}
\end{itemize}
Let $\calJ$ be the set of all the jobs. Further, we set:
\begin{align*}
\mathcal{F} = \{\mathcal{F}_s\colon s\in S\}; \quad\quad \mathcal{R} = \bigcup_{s \in S} \mathcal{R}_s; \quad\quad \mathcal{P} = \bigcup_{s \in S} \mathcal{P}_s;
\end{align*} 


Additionally, we introduce $\mu$ jobs, $\calX = \{ X_1, \ldots, X_\mu \}$, each having a unit length, and a job $L$ of length $3 s \mu^3 s_T$ (thus, the length of $L$ is larger than the length of all the jobs of  $\calJ \setminus \{L\}$). 

There are $\mu s_\Sigma$ agents in total. For each number $s \in S$ we introduce $s \mu$ agents. Let $\calA_s$ be the set of these agents. We partition $\calA_s$ into $\mu$ sets of $s$ agents: $\calA_{s,1},\dots,\calA_{s,\mu}$. Figure~\ref{nb_late_tasks} 
represents the preferred schedule of the $j$-th agent from $\calA_{s,i}$ ($i\in [\mu], ~j\in[s]$).
For all the agents, job $X_i$ ($i\in [\mu]$) starts at time $i s_T + i-1$, and job $L$ starts at time $D=s_\Sigma+\mu=\mu(s_T+1)$.
Further, for all the agents of $\calA_{s,i}$, job $F_s$ starts at time $i s_T + (i -1) - s$ (i.e., for these agents, $F_s$ is scheduled just before job $X_i$). Further, in this schedule job $R_{s,i,j}$ is put just before job $F_s$: at time $(i-1) (s_T+1)$, and job $P_{s,i,j,k}$ ($k\in [\mu-1]$) is scheduled at time $(k-1) (s_T+1)$ if $k<i$, and at time $k (s_T+1)$ if $k\geq i$. 
All the other jobs are scheduled after job $L$, i.e., at soonest at time $D+p_L$. 
Let us arbitrarily label the agents from $0$ to $\mu s_\Sigma-1$. The jobs of Agent $i$ which are not already scheduled before $D+p_L$ are scheduled in an arbitrarily order after $D+p_L$, except that the $2\mu-2$ latest jobs of the schedule are the jobs of $\calP$ which are scheduled before $D$ in the preferred schedule of Agent $(i+1 \mod \mu s_\Sigma)$, followed by the jobs of $\calP$ which are scheduled before $D$ in the preferred schedule of Agent $(i+2 \mod \mu s_\Sigma)$. This will ensure that each job of $\calP$ appears only twice in the $(2\mu-2)$ last jobs of the agents (since, for each job of $\calP$, only one agent schedules it before $D$).

\begin{figure*}[t!]
\centering
\includegraphics[width=\columnwidth]{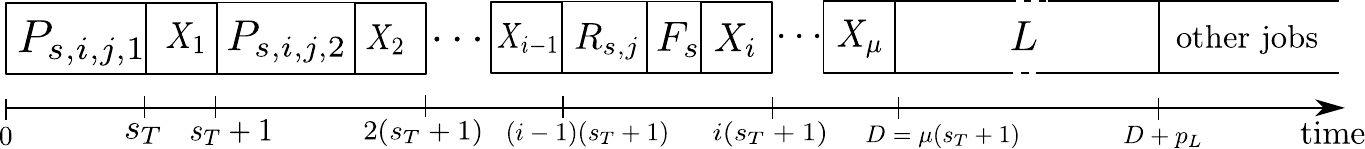}
\caption{Preferred schedule of the $j$-th agent of $\calA_{s,i}$.}
\label{nb_late_tasks}
\end{figure*}

We will now show that the answer to the 3-partition problem on instance $I$ is ``yes" if and only if the optimal schedule for $\Sigma$-$U$ on $I'$ starts as follows: $\Big(\calF_1, X_1, \calF_2, X_2, \calF_\mu, X_\mu \Big)$, where each set $\calF_i$ consists of jobs from $\calF$ with lengths summing up to $s_T$.

If the schedule for $\Sigma$-$U$ on $I'$ starts as follows: $\Big(\calF_1, X_1, \calF_2, X_2, \calF_\mu, X_\mu \Big)$, where each set $\calF_i$ ($i \in [\mu]$) consists of jobs from $\calF$ with lengths summing up to $s_T$, then the solution of 3-partition is ``yes'' since each job of $\calF$ has the length of a number of $S$. Let us now assume that the solution of 3-partition on Instance $I$ is ``yes''. We will show that the optimal solution of $\Sigma$-$U$ on Instance $I'$ indeed starts with: $\Big(\calF_1, X_1, \calF_2, X_2, \calF_\mu, X_\mu \Big)$.


Let us consider an optimal schedule $\sigma^*$ for $I'$. First, we will show that in $\sigma^*$, each job $X_i$ ($i\in [\mu]$) is scheduled at latest at time $i s_T + (i-1)$. 

Indeed, in the preferred schedules of all the agents, $X_i$ is completed at time  $i s_T + (i-1)$. If in $\sigma^*$, $X_i$ would not be completed at latest at time  $i s_T + (i-1)$, then it would mean that another job $T\in (\calJ - \calX)$ is scheduled before $X_i$. In this case, swapping $T$ and $X_i$ would not increase the number of late jobs. This is the case because scheduling $X_i$ before $i s_T + (i-1)$ decreases by $\mu s_\Sigma$ (this is the number of agents) the number of late jobs, and the length of $X_i$ is smaller than the one of $T$ so the swap of $X_i$ and $T$ will not delay jobs other than $T$.   

Second, we will show that in $\sigma^*$, job $L$ is scheduled in the last position. For this we consider two cases:
\begin{enumerate}
\item Let us first consider what happens if this job is not late, i.e., if it is scheduled in $\sigma^*$ at latest at time $D$. Let us now look at the $2\mu-2$ last jobs of $\sigma^*$. Each of these jobs is late for  at least $\mu s_\Sigma -2$ agents (all the agents except two), since for each job of $\calP$, only two agents have it in one of their $2\mu-2$ last positions of their preferred schedule (and all these jobs are of length $s_T$). Thus the total number of late jobs is at least $(\mu s_\Sigma - 2)(2\mu-2)$.

\item Let us now consider the case where job $L$ is late in $\sigma^*$: it is scheduled after time $D$. In this case, it will be late for all the agents, so we can assume that it is scheduled in the last position of $\sigma^*$. 
Thus, all the jobs of $\calJ\setminus \{L\}$ are scheduled before $D+p_L$ in $\sigma^*$ (this is true since the length of $L$ is larger that the total length of the jobs of $\calJ \setminus \{L\}$). Since each job of $\calP\cup\calR$ appears only once before $D+p_L$ in the preferred schedules of the agents, each job of $\calP\cup\calR$ will be late for at most one agent: the number of jobs of $\calP\cup\calR$ which will be late is thus at most the number of jobs of $\calP\cup \calR$: $\mu^2 s_\Sigma$. The number of jobs of $\calF$ which are scheduled before $D+p_L$ in the preferred schedules of the agents is $\mu s_\Sigma$ (indeed for each $s\in S$, job $F_s$ appears for $s$ agents just before job $X_i$, with $i\in[\mu]$). Thus, the number of jobs from $\calF$ which will be late in $\sigma^*$ is at most equal to $\mu s_\Sigma$. Job $L$ is late for all the $\mu s_\Sigma$ agents, and we have already seen that the jobs of $\calX$ are not late in $\sigma^*$. Therefore, the total number of jobs which will be late in $\sigma^*$ is at most $(\mu +2)\mu s_\Sigma$. This is smaller that the lower bound of the number of jobs late if $L$ is not late in $\sigma^*$ (this lower bound was $(\mu s_\Sigma - 2)(2\mu-2)$, and we have assumed that $\mu>4$).
\end{enumerate}
From the case analysis we conclude that in $\sigma^*$, $L$ is scheduled in the last position.

Third, we infer that the jobs of $\calP\cup\calR$ are scheduled in $\sigma^*$ at soonest at time $D$. If this was not the case, a job of $\calF$ would be scheduled after time $D$. We argue that by swapping this job with a job of $\calP\cup\calR$ scheduled before $D$ we would not increase the number of late jobs. Indeed, if a job of $\calP\cup\calR$ is completed at latest at time $D$ then it will be, in the best case, scheduled on time for all the agents, whereas if it is completed after time $D$ (but before $D+p_L$) it will be late for (only) one agent. Moving forward a job $F_s\in\calF$ which is completed after time $D$ so that it is now completed at latest at time $D-1$ ($F_s$ is shorter that any job of $\calP\cup\calR$) will decrease the number of late jobs by at least $s\geq 1$ since $s$ agents have this job completed at time $D-1$ in their preferred schedules. Since $F_s$ is shorter than any job of $\calP\cup\calR$, then doing such a swap does not make any other job late.

We have seen that in $\sigma^*$, the jobs scheduled before $D$ are the jobs of $\calX\cup \calF$, and that the jobs of $\calX$ are not late (i.e. $X_i$ is scheduled at latest at time $i s_T + i-1$). Let us now see how these jobs are scheduled. 

Recall that for each $s\in S$, job $F_s$ is completed for $s$ agents at time $s_T$, for $s$ agents at time $2s_T+1$, and so forth (for each $i\in[\mu]$ it is completed $s$ times in $i s_T + (i-1)$). Therefore, if, in $\sigma^*$, job $F_s$ is completed at latest at time $s_T$, is will not be late for any agent; if it is completed after $s_T$ but at latest at time $2 s_T +1$, it will be late for $s$ agents, and so forth. For each $i\in [\mu]$ we define Slot $i$ as the time interval $[(i-1)s_T + i-1, is_T + i-1 ]$. As we have seen, if $F_s$ is completed in Slot $i$ ($i\in [\mu]$), it will be late for $(i-1) s$ agents. Since the length of $F_s$ is $s$, it means that \emph{each unit} of a job $F \in \calF$ such that $F$ is completed in Slot $i$ adds $(i-1)$ to the number of late jobs. Let $U_i$ be the sum of the lengths of jobs of $\calF$ completed in Slot $i$. The number of late jobs of $\calF$ is then $\sum_{i=1}^\mu U_i (i-1)$. 
Therefore, to minimize the number of late jobs, $U_1$---the total length of jobs completed in Slot 1---should be as large as possible, and then, for the remaining jobs, $U_2$ should be as large as possible, and so forth. Since $X_i$ has to be scheduled at latest at time $i s_T + (i-1)$, i.e. at the beginning of Slot $i+1$, the number of late jobs is minimized if $X_i$ is scheduled exactly at time $i s_T + (i-1)$. This can be done only if it is possible to schedule the jobs of $\calF$ in the slots between the jobs of $\calX$. Since we have assumed that there is a ``yes" solution to the 3-partition problem on instance $I$, it is possible to partition the jobs of $\calF$ in triples with lengths summing up to $s_T$. Thus, each of this triple will correspond to a triple of jobs in the same slot. Hence the optimal solution of $\Sigma$-$U$ starts as follows: $\Big(\calF_1, X_1, \calF_2, X_2, \calF_\mu, X_\mu \Big)$, where each set $\calF_i$ consists of jobs from $\calF$ with lengths summing up to $s_T$. This completes the proof. 
\end{proof}


Nonetheless, if the jobs have the same size, the problem can be solved in polynomial time (highlighting the additional complexity brought by the main element of the collective scheduling). Our proof uses the idea of~\citet{dwork2001rank} who proved an analogous result for the Spearman distance.

\begin{proposition}\label{thm:tardiness_sizes_the_same}
If all jobs have the same size, for each delay cost $f \in \{T, U, L, E, D, \mathit{SD}\}$ rule $\sum$-$f$ can be computed in polynomial time.
\end{proposition}
\begin{proof}
    Let us fix $f \in \{T, U, L, E, D, \mathit{SD}\}$. We reduce the problem of finding a collective schedule to the assignment problem. Observe that when the jobs have all the same size, say $p$, then in the optimal schedule each job should be started at time $\ell p$ for some $\ell \in \{0, \dots, m-1 \}$. Thus, we construct a bipartite graph where the vertices on one side correspond to $m$ jobs and the vertices on the other side to $m$ possible starting times of these jobs. The edge between a job $J$ and a starting time $\ell p$ has a cost which is equal to the total cost caused by job $J$ being scheduled to start at time $\ell p$.
  The cost can be computed independently of how the other jobs are scheduled, and is equal to $\sum_{a \in N} f(\ell p + 1, C_i(\sigma_a) \big)$. Thus, a schedule that minimizes the total cost corresponds to an optimal assignment of $m$ jobs to their $m$ slots. Such an assignment can be found in polynomial time, e.g., by the Hungarian algorithm.   
\end{proof}

We conclude this section by observing that hardness of computing $\sum$-$K$ and $\sum$-$S$ rules can be deduced from the hardness of computing Kemeny rankings~\cite{dwork2001rank}. 
\begin{proposition}
Computing $\sum$-$K$ and $\sum$-$S$ is $\np$-hard even for $n=4$ agents and when all jobs have the same unit size.
\end{proposition}

\subsection{$L_p$-norm of Delay Costs, $p>1$}

We start by observing that the general case is hard even for two agents. The proof of the below theorem works also for $p = \infty$, i.e., for $\max$-$\{ T, E, D \}$. 

\begin{theorem}\label{thm:L_p}
  For each $p>1$, finding a schedule returned by $L_p$-$\{ T, E, D \}$ is $\np$-hard, even for two agents. 
\end{theorem}

\begin{proof}
  Let us fix $p > 1$. We show a reduction from \textsc{Partition}.
  In \textsc{Partition} we are given a set of integers $S = \{ s_1, \ldots, s_n\}$, $s_i < s_{i+1}$,  and we ask whether $S$ can be partitioned in two sets $S_a$ and $S_b$ that have the same sum, $s = \nicefrac{1}{2} \cdot \sum s_i$.

  We construct an instance of the problem of finding an optimal collective schedule according to $L_p$-$T$ as follows (our construction is inspired by~\citet{agnetis2004scheduling}).
  For each $s_i \in S$ we introduce two jobs $J_i^{(a)}$ and $J_i^{(b)}$, both with length $p_i^{(a)}=p_i^{(b)}=s_i$.
  We have two agents, $a$ and $b$.
  Both agents prefer a schedule executing jobs in order of their increasing lengths (an SPT schedule).
  For each pair of jobs $(J_i^{(a)}, J_i^{(b)})$ with equal lengths, agent $a$ prefers $J_i^{(a)}$ to $J_i^{(b)}$, while agent $b$ prefers $J_i^{(b)}$ to $J_i^{(a)}$.
  Thus, the preferred schedule $\sigma_a$ of agent $a$ is $(J_1^{(a)}, J_1^{(b)}, J_2^{(a)}, J_2^{(b)}, \dots, J_n^{(a)}, J_n^{(b)})$; while $\sigma_b$ is $(J_1^{(b)}, J_1^{(a)}, J_2^{(b)}, J_2^{(a)}, \dots, J_n^{(b)}, J_n^{(a)})$. We ask whether there exist a schedule with a cost of $\sqrt[p]{2 s^p}$.
  
  Assume there exists a partition of $S$ into two disjoint sets $S_a, S_b$ where $\sum_{v \in S_a}v =\sum{v \in S_b}v = s$.
  We construct an SPT schedule $\sigma=( J_1^{(\cdot)}, J_1^{(\cdot)}, J_2^{(\cdot)}, J_2^{(\cdot)}, \dots, J_n^{(\cdot)}, J_n^{(\cdot)} )$.
  For each pair of jobs $\{ J_i^{(a)}), J_i^{(b)}) \}$ of equal length, if $s_i \in S_a$,
  the jobs are scheduled in order $(J_i^{(a)}, J_i^{(b)})$;
  otherwise (i.e., $s_i \in S_b$) in order $(J_i^{(b)}, J_i^{(a)})$.
  For each pair of jobs, the order $(J_i^{(a)}, J_i^{(b)})$ increases agent $b$'s tardiness by $p_i^{(a)}=s_i$; while agent $a$'s tardiness is not increased. Similarly, the order $(J_i^{(b)}, J_i^{(a)})$ increases agent $a$'s tardiness by $p_i^{(b)}=s_i$. Consequently, $T(\sigma, \sigma_a) = S_b = s$ and $T(\sigma, \sigma_b) = S_a = s$, and thus the total cost is $\sqrt[p]{2 s^p}$.

  Assume there is a schedule $\sigma$ with the total cost of $\sqrt[p]{2 s^p}$.
  We first show it has to be an SPT schedule.
  For the sake of contradiction, assume that a non-SPT schedule $\sigma'$ has the minimal cost.
  Pick two jobs $J_i$, $J_j$ scheduled in a non-SPT order in $\sigma'$:
  $J_i$ is scheduled before $J_j$, but $p_i > p_j$. 
  If we switch the order of jobs, the jobs $J_k$ executed between $J_j$ and $J_i$ complete earlier.
  Moreover, as both agents prefer $J_j$ to $J_i$,
  the T tardiness measure drops ($J_k$ are less late in the switched schedule).
  Thus, the switched schedule has a lower cost, which contradicts the assumption that $\sigma'$ is optimal.

  Next, observe that the sum of agents' tardiness measures is at least $2 s$.
  Consider an SPT schedule. For each pair of jobs $(J_i^{(a)}, J_i^{(b)})$,
  if $J_i^{(a)}$ is scheduled before $J_i^{(b)}$,
  $T(\sigma, \sigma_b)$ is increased by $p_i$;
  otherwise, $T(\sigma, \sigma_a)$ is increased by $p_i$.
  As $T(\sigma, \sigma_a) + T(\sigma, \sigma_b) = 2s$, by the convexity of the $L_p$-norm (with $p>1$),
  a schedule with the total cost of $\sqrt[p]{2 s^p}$ has to be a schedule with the minimal cost. Further, this cost is equal to $\sqrt[p]{2 s^p}$, if and only if the total tardiness of agents $a$ and $b$ are equal. 

  Now, observe that the order of a pair of jobs with equal length $\{ J_i^{(a)}), J_i^{(b)}) \}$ defines thus the partition: 
  the order $(J_i^{(b)}, J_i^{(a)})$ corresponds to $s_i$ in $S_a$;
  the order $(J_i^{(a)}, J_i^{(b)})$ corresponds to $s_i$ in $S_b$.
  As $T(\sigma, \sigma_a) = T(\sigma, \sigma_b) = s$, 
  $\sum_{s_i \in S_a} s_i = \sum_{s_i \in S_b} s_i = s$.  
\end{proof}

Moreover, as shown below, $\max$-$\{T, E, D, \mathit{SD}\}$ is $\np$-hard even for unit-size jobs.

\begin{theorem}\label{thm:max-t-np-hard}
For each delay cost $f \in \{T, E, D, \mathit{SD}\}$, finding a schedule returned by $\max$-$f$ is $\np$-hard, even for unit-size jobs. 
\end{theorem}
\begin{proof}
We reduce from the \textsc{ClosestString}, which is $\np$-hard even for the binary alphabet. Let $I$ be an instance of \textsc{ClosestString} with the binary alphabet. In $I$ we are given a set of $n$ 0/1 strings, each of length $m$, and an integer $d$; we ask if there exists a ``central string'' with the maximum Hamming distance to the input strings no greater than $d$.

From $I$ we construct an instance $I'$ of $\max$-$f$ collective schedule in the following way. We have $2m$ jobs: for each $i \in [m]$ we introduce two jobs, $J_i^{(a)}$ and $J_i^{(b)}$. For each input string $s$ we introduce one agent: the agent puts a job $J_i^{(\cdot)}$ before $J_j^{(\cdot)}$ in her preferred schedule whenever $i<j$. Further, she puts $J_i^{(a)}$ before $J_i^{(b)}$ if $s$ has ``one'' in the $i$-th position and $J_i^{(b)}$ before $J_i^{(a)}$, otherwise.

Let us call a schedule where $J_i^{(\cdot)}$ is put before $J_j^{(\cdot)}$ whenever $i<j$, a \emph{regular schedule}.
We consider the schedule $\sigma^*$ returned by $\max$-$f$, and we show that this schedule is regular (or that it can be transformed into a regular schedule of the same cost). Let us consider that there is in $\sigma^*$ two jobs $J_i^{(\cdot)}$ and $J_j^{(\cdot)}$ such that $J_j^{(\cdot)}$ is scheduled before  $J_i^{(\cdot)}$ whereas $i<j$. 
Swapping $J_j^{(\cdot)}$ with $J_i^{(\cdot)}$
changes only $J_j^{(\cdot)}$ and $J_i^{(\cdot)}$ completion times (as jobs are unit-size).
By case analysis on both jobs' positions relative to $2i$ and $2j$ (6 cases, as $j$ is before $i$), 
for any $f \in \{T, E, D, \mathit{SD}\}$, 
swapping these jobs
does not
increase $f$.
Thus, if $\sigma^*$ is not regular, we can transform it into a regular schedule as follows: by swapping $J_1^{(\cdot)}$ with another job $J_k^{(\cdot)}$ (if $J_1^{(\cdot)}$ is not at position $1$ or $2$, whereas  $J_k^{(\cdot)}$, with $k>1$, is at one of these positions), we do not increase the cost $f$ of the schedule, and thus we obtain a schedule where the jobs $J_1^{(\cdot)}$ are at their regular positions. We continue with at most $2m$ such swaps for the remaining positions $i\in[m]$, ending up with a regular schedule.

Let us now consider that $f=T$ (resp. $f=E$). Observe that if we put $J_i^{(a)}$ before  $J_i^{(b)}$ in a regular schedule, then we increase the tardiness (resp. earliness) of each agent having ``zero'' in the $i$-th position by one. Conversely, if we schedule $J_i^{(b)}$ before  $J_i^{(a)}$, then we increase the tardiness (resp. earliness) of agents having ``one'' in the $i$-th position by one. Thus, a (regular) collective schedule corresponds to a ``central string'': $J_i^{(a)}$ scheduled before  $J_i^{(b)}$ in a collective schedule corresponds to a central string having ``one'' in the $i$-th position, and $J_i^{(b)}$ scheduled before  $J_i^{(a)}$, corresponds to ``zero''. With such interpretation, the $\max$-$T$ (resp. $\max$-$E$) of a regular schedule is simply the maximum Hamming distance to the input strings.
Consequently, we get that the answer to the initial instance $I$ is ``yes'', iff the optimal solution for $I'$ is a schedule with  $\max$-$T$ (resp. $\max$-$E$) not larger than $d$.

When $f=D$ (resp. $f=\mathit{SD}$), the principle of the proof is the same: $J_i^{(a)}$ before  $J_i^{(b)}$ in a regular schedule increases the deviation (resp. squared deviation) of each agent having ``zero'' in the $i$-th position by two. Conversely, if we schedule $J_i^{(b)}$ before  $J_i^{(a)}$, then we increase the deviation (resp. squared deviation) of agents having ``one'' in the $i$-th position by one. 
Consequently, we get that the answer to the initial instance $I$ is ``yes'', iff the optimal solution for $I'$ is a schedule with  $\max$-$D$ (resp. $\max$-$\mathit{SD}$) not larger than $2d$.
\end{proof}

\section{Experimental Evaluation}
The goal of our experimental evaluation is, first, to demonstrate that, while most of the problems are NP-hard, an Integer Linear Programming (ILP) solver finds optimal solutions for instances with reasonable sizes. Second, to quantitatively characterize the impact of collective scheduling compared to the base social choice methods. Third, to compare schedules built with different approaches (cost functions and axioms). 
We use tardiness $T$ as a representative cost function: it is NP-hard in both $\Sigma$ and $\max$ aggregations; and easy to interpret.

\medskip\noindent
\textbf{Settings.} A single \emph{experimental scenario} is described by a profile with preferred schedules of the agents  and by a maximum length of a job $p_{\max}$. We instantiate the preferred schedules of agents using PrefLib~\citep{conf/aldt/MatteiW13}.
We treat PrefLib's candidates as jobs.
We use datasets where the agents have strict preferences over all candidates.
We restrict to datasets with both large number of candidates and large number of agents: we take two datasets on AGH course selection (\textsc{agh1} with 9 candidates and 146 agents; and \textsc{agh2} with 7 candidates and 153 agents) and \textsc{sushi} dataset with 10 candidates and 5000 agents. Additionally, we generate preferences using the \citet{mallowModel} model (\textsc{mallows}) and Impartial Culture (\textsc{impartial}), both with 10 candidates and 500 agents.
We use three different values for $p_{\max}$: $10$, $20$ and $50$.
For each experimental scenario we generate 100 instances---in each instance pick the lengths of the jobs uniformly at random between 1 and $p_{\max}$ (in separate series of experiments we used exponential and normal distributions; we found similar trends to the ones discussed below). For each scenario, we present averages and standard deviations over these 100 instances.

\medskip\noindent
\textbf{Computing Optimal Solutions.}
We use standard ILP encoding:
for each pair of jobs $(i,j)$, we introduce two binary variables $prec_{i,j}$ and $prec_{j,i}$ denoting precedence: $prec_{i,j} = 1$ iff $i$ precedes $j$ in the schedule.
($prec_{i,j} + prec_{j,i} = 1$ and, to guarantee transitivity of $prec$, for each triple $i,j,k$, we have $prec_{i,j} + prec_{j,k} - prec_{i,k} \leq 1$).
We run Gurobi solver on a 6-core (12-thread) PC. An \textsc{agh} instance takes, on the average, less than a second to solve, while a \textsc{sushi} instance takes roughly 20 seconds.
In a separate series of experiments, we analyze the runtime on \textsc{impartial} instances as a function of number of jobs and number of voters. A 20 jobs, 500 voters instance  with $\sum$-$T$ goal takes 8 seconds; while a $\max$-$T$ goal takes two minutes. A 10 jobs, 5000 voters takes 8 seconds with $\sum$-$T$ goal and 28 seconds with $\max$-$T$ goal. Finally, 20 jobs, 5000 voters take 23 seconds for with $\sum$-$T$ and 20 minutes with $\max$-$T$.
For 30 jobs, the solver does not finish in 60 minutes.
Running times depend thus primarily on the number of jobs and on the goal.
We conclude that, while the problem is strongly $\np$-hard,
it can be solved in practice for thousands of voters and up to 20 jobs.
We consider these running times to be satisfactory:
first, for a population it might be difficult to meaningfully express preferences for dozens of jobs~\cite{Mil56} (therefore, the decision maker would probably combine jobs before eliciting preferences); second, gathering preferences takes non-negligible time; and, finally, in our motivating examples (public works, lecture hall) individual jobs last hours to weeks.


\begin{figure}[t!]
\begin{minipage}[h]{0.45\linewidth}
  \centering
  \includegraphics[width=\textwidth]{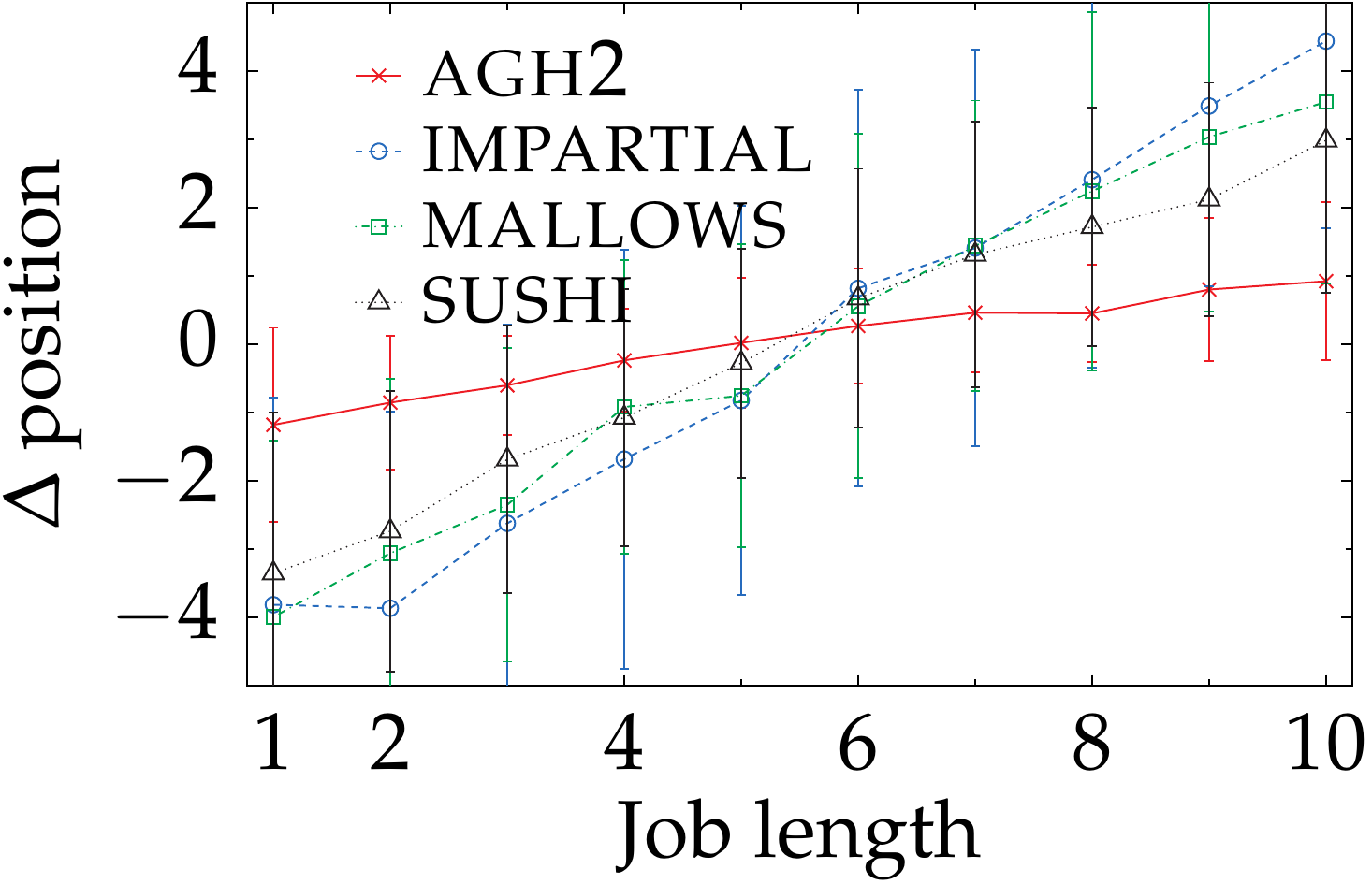}
  \vspace{-1em}
  {\Large $~~~~\Sigma$-$T$ } \medskip
\end{minipage}
\hspace{0.1cm}
\begin{minipage}[h]{0.45\linewidth}
  \centering
  \includegraphics[width=\textwidth]{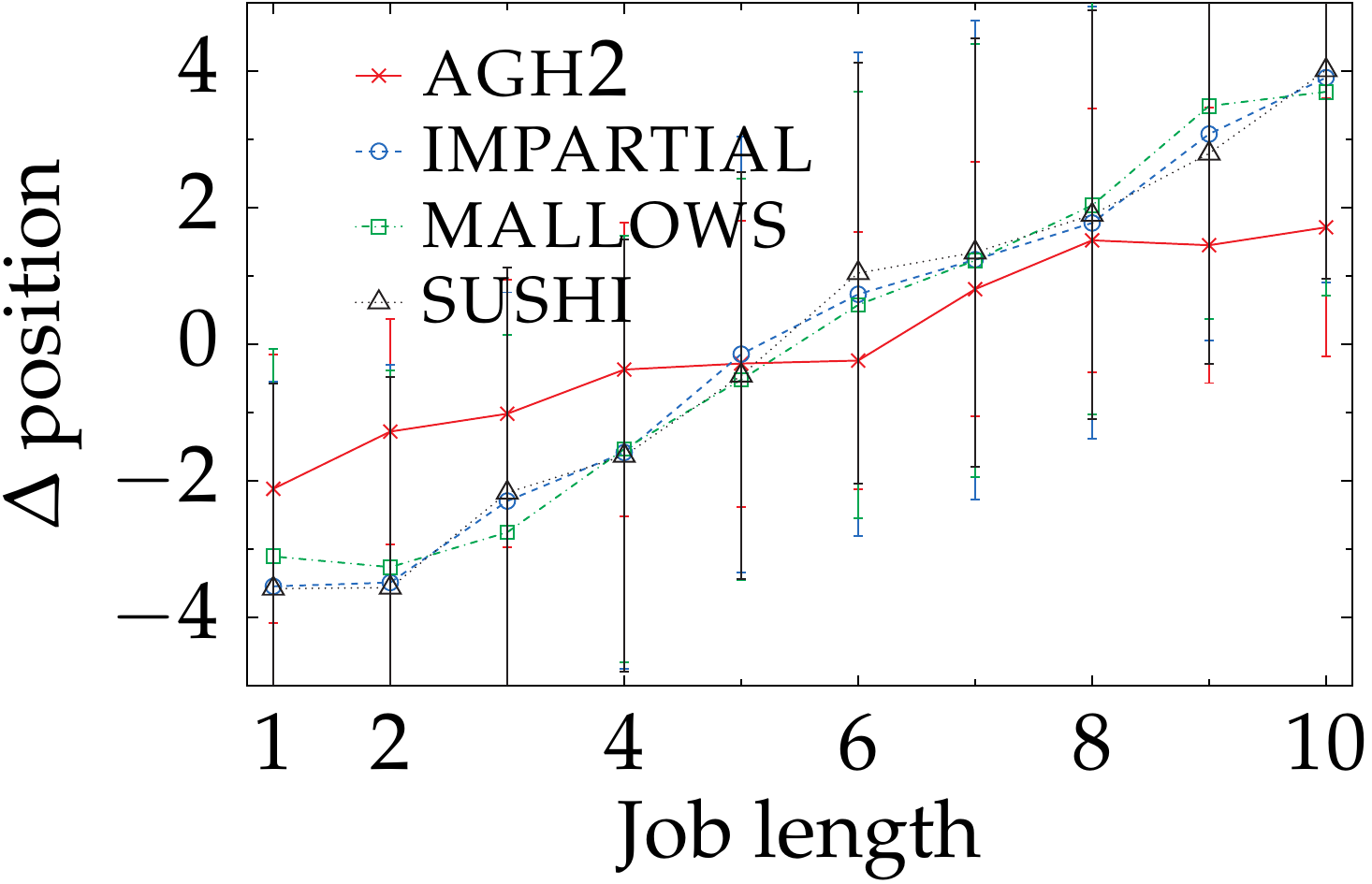}
  \vspace{-1em}
   {\Large $~~~~~\max$-$T$ } \medskip
\end{minipage}

\caption{The average change in jobs' position. A point $(x, y)$ in the plot denotes that a job of length $x$ is on the average scheduled by $y$ positions later than when we ignore jobs' durations. $p_{\max}=10$ ($p_{\max}=20$ and $p_{\max}=50$ show very similar trends.)
}
\label{fig:job_advancement}
\end{figure}

\medskip\noindent
\textbf{Analysis of the Results.}
First, we analyze job's rank as a function of its length.
We compute a reference collective schedule for an instance with the same agents' preferences, but unit-size jobs (it thus corresponds to the classic preference aggregation problem with $\Sigma$-$T$ or $\max$-$T$ goal).
We then compute and analyze the collective schedules. 
Over 100 instances, as jobs' durations are assigned randomly, all the jobs' durations should be in the preferred schedules in, roughly, all positions.
Thus, \emph{on the average}, short jobs should be executed earlier, and long jobs later than in the reference schedule (in contrast, in any single experiment, if a large majority puts a short job at the end of their preferred schedules, the job is not automatically advanced).
To confirm this hypothesis, for each instance and each job we compare its position to the position in the reference schedule. Figure~\ref{fig:job_advancement} shows the average position change as a function of the job lengths.
In collective schedules, short jobs (e.g., of size 1) are advanced, on the average, 2-4 positions in the schedule, compared to schedules corresponding to the standard preference aggregation problem.
The experiments thus confirm that the lengths of the jobs have profound impact on the schedule.

\begin{table}[t!]
  \centering
  \begin{tabular}{l|c|c|c|c|c|}
    \multirow{2}{*}{Dataset} & \multicolumn{2}{c|}{PTA C. Paradox} & \multicolumn{2}{c|}{PTA Copeland $\nicefrac{\cdot}{\cdot}$} & \multirow{2}{*}{$\Delta$Gini} \\
    & $\Sigma$-$T$ & $\max$-$T$ & $\Sigma$-$T$ & $\max$-$T$ & \\
    \hline
\textsc{agh1} & 6\% & 15\% & 1.03 & 1.23 & 0.07 \\
\textsc{agh2} & 5\% & 18\% & 1.03 & 1.28 & 0.12 \\
\textsc{sushi} & 7\% & 24\% & 1.02 & 1.22 & 0.06 \\
\textsc{impartial} & 3\% & 8\% & 1.00 & 1.01 & 0.00 \\
\textsc{mallows} & 10\% & 24\% & 1.03 & 1.21 & 0.08
  \end{tabular}
  \caption{``PTA C. Paradox'' gives the mean frequencies of violating the PTA Condorcet principle for optimal solutions for $\Sigma$-$T$ and $\max$-$T$. ``PTA Copeland $\nicefrac{\cdot}{\cdot}$'' denotes the ratio of sum/max $T$ for PTA Copeland's schedule to their optimums. ``$\Delta$Gini'' shows the average of differences in the Gini indices: Gini($\max$-$T$) - Gini($\Sigma$-$T$).}
  \label{tab:exp_summary}
\end{table}

Second, we check how frequent are PTA-Condorcet paradoxes.  
For each instance, we counted how many out of ${m \choose 2}$ job pairs are scheduled in a non-PTA-Condorcet consistent order. Table~\ref{tab:exp_summary} shows that both $\Sigma$-$T$ and $\max$-$T$ often violate the PTA Condorcet principle.
Table~\ref{tab:exp_summary} also shows the average ratio between the ($\Sigma$ and $\max$) tardiness of schedules returned by the PTA Copeland's rule, and the tardiness of optimal corresponding schedules.
These ratios are small: roughly 3\% degradation for $\Sigma$ and 24\% for $\max$. Thus, though PTA Copeland's rule does not explicitly optimize $\max$-$T$ and
$\Sigma$-$T$, on average, it returns schedules close to the optimal for these criteria.


Third, we analyze how fair are $\Sigma$-$T$ and $\max$-$T$.
We analyzed Gini indices of the vectors of agents' tardiness.
Table~\ref{tab:exp_summary} shows that, interestingly, $\Sigma$-$T$ is more fair (smaller average Gini index), even though $\max$-$T$ seemingly cares more about less satisfied agents. Yet, the focus of $\max$-$T$ on the worst-off agent makes it effectively ignore all the remaining agents, increasing the societal inequality.

\section{Discussion and Conclusions}\label{sec:conclusions}
The principal contribution of this paper is conceptual---we introduce the notion of the \emph{collective schedule}.
We believe that collective scheduling addresses natural problems involving jobs or events having diverse impacts on the society.
Such problems do not fit well into existing scheduling models.
We demonstrated how to formalize the notion of the collective schedule
by extending well-known methods from social choice.
While collective scheduling is closely related to preference aggregation,
these methods have to be extended to take into account lengths of jobs.
Notably, we proposed to judge the quality of a collective schedule by comparing the jobs' completion times between the collective and the agents' preferred schedules.
We also showed how to extend the Condorcet principle to take into account lengths of jobs.


We conclude that there is no clear winner among the proposed scheduling mechanisms. Similarly, in the classic voting, there is no clear consensus regarding which voting mechanism is the best. For example, we showed that the comparison of the cost-based and PTA-Condorcet-based scheduling exposes a tradeoff between reinforcement and the PTA Condorcet principle. Thus, the question which mechanism to choose is, for example, influenced by the subjective assessment of the mechanism designer with respect to which one of the two properties she considers more important. 

Our main conclusion from the theoretical analysis of computational complexity and from the experimental analysis is that using cost-based scheduling methods is feasible only if the sizes of the input instances are moderate (though, these instances may represent many realistic situations). In contrast, PTA Condorcet-based methods are feasible even for large instances.
We drew a boundary between $\np$-hard and polynomial-time solvable problems. In several cases, problems become $\np$-hard with non-unit jobs, therefore showing additional complexity stemming from scheduling, as opposed to standard voting. Moreover, our experiments suggest that there is a clearly visible difference between schedules returned by different methods of collective scheduling.

Both scheduling and social choice are well-developed fields with a plethora of models, methods and results.
It is natural to consider more complex scheduling models in the context of collective scheduling, such as processing several jobs simultaneously (multiple processors with sequential or parallel jobs),  jobs with different release dates or dependencies between jobs.
Each of these extensions raises new questions on computability/approximability of collective schedules.
Another interesting direction is to derive desired properties of collective schedules (distinct from PTA-Condorcet), and then formulate scheduling algorithms satisfying them.

\subsection*{Acknowledgments}
  This research has been partly supported by the Polish National
  Science Center grant Sonata (UMO-2012/07/D/ST6/02440), a Polonium
  grant (joint programme of the French Ministry of Foreign Affairs,
  the Ministry of Science and Higher Education and the Polish Ministry
  of Science and Higher Education) and project TOTAL that has
  received funding from the European Research Council (ERC) under the
  European Union's Horizon 2020 research and innovation programme
  (grant agreement No 677651). 

  Piotr Skowron was also supported by a Humboldt fellowship for postdoctoral researchers.

\bibliographystyle{abbrvnat}
\bibliography{collsched}

\end{document}